\documentclass[12pt, reco]{article}
\usepackage{amssymb,amscd,amsmath,amsthm}

\usepackage[colorinlistoftodos]{todonotes}
\usepackage[colorlinks=true, allcolors=blue]{hyperref}
\usepackage{lmodern}
\usepackage[english]{babel}
\usepackage{color}
\usepackage{bm}
\usepackage{graphicx}
\usepackage{tikz}
\usepackage[margin=1in]{geometry}
\newtheorem{theorem}{Theorem}[section]
\newtheorem{proposition}[theorem]{Proposition}

\newtheorem{lemma}[theorem]{Lemma}
\newtheorem{definition}[theorem]{Definition}
\newtheorem{example}[theorem]{Example}
\newtheorem{remark}[theorem]{Remark}

 \usepackage{algpseudocode}
 \usepackage{algorithmicx}
 \usepackage{algorithm}
\title{Quantum Viterbi Algorithm}

\begin{document}

\maketitle

\centerline{\author{\Large Luigi Accardi }}\centerline{ Centro Vito Volterra, Universit\`a  di Roma ''Tor Vergata''}\centerline{ Roma I-00133, Italy }\centerline{accardi@volterra.uniroma2.it}\vskip0.3cm

\centerline{ \author{\Large Abdessatar Souissi}}\centerline{Department of Management Information Systems, College of Business and Economics, }\centerline{Qassim University, Buraydah 51452, Saudi Arabia}\centerline{\textit{a.souaissi@qu.edu.sa}}\vskip0.3cm

\centerline{\author{\Large{ El Gheteb Soueidi}}}\centerline{Department of Mathematics and Informatics, Faculty of Sciences and Technologies,} \centerline{University of Nouakchott, Mauritania} \centerline{elkotobmedsalem@fst.e-una.mr}\vskip0.3cm

\centerline{\author{\Large{ Farrukh Mukhamedov}}}\centerline{Department of Mathematical Sciences, College of Science, United Arab Emirates University,} \centerline{ P.O.
Box 15551, Al Ain, Abu Dhabi, UAE} \centerline{farrukh.m@uaeu.ac.ae}\vskip0.3cm

\centerline{\author{\Large Mohamed Rhaima}}\centerline{Department of Statistics and Operations Research, College of Sciences,}\centerline{ King Saud University, P. O Box 2455, Riyadh 11451, Saudi Arabia}\centerline{mrhaima.c@ksu.edu.sa}

\tableofcontents

 \begin{abstract}
We introduce a  quantum Viterbi decoding algorithm  for hidden quantum Markov models (HQMMs) motivated by quantum information processing and quantum algorithms. Given a finite sequence of measurement outcomes, the algorithm identifies hidden quantum trajectories that maximize a joint decoding functional, serving as a genuine quantum analogue of the classical Viterbi score. Unlike classical hidden Markov models, where decoding optimizes over a finite discrete state space, our method performs optimization over a continuous manifold of pure quantum effects, thereby exploiting coherent superpositions in the hidden memory. We prove a strict quantum advantage: coherent hidden trajectories can achieve decoding scores that strictly exceed any classical strategy constrained to diagonal (commuting) effects, even when both models share the same observed statistics. These results position quantum Viterbi decoding as a concrete quantum algorithmic primitive for sequential decision-making, with direct applications to quantum memories, quantum communication with memory, and near-term quantum machine learning on NISQ devices.

\end{abstract}

\section{Introduction}
For over five decades, the Hidden Markov Model (HMM) has stood as a quiet workhorse of artificial intelligence—a testament to the power of modeling the world through what we cannot see. From the earliest speech recognition systems that learned to parse the cacophony of human speech into discrete phonetic units, to modern computational biology where these models trace the intricate architecture of genes within vast strands of DNA, the HMM has provided a unifying framework for understanding sequential data \cite{Rab2002, Eddy1996, Mor2021}. Its enduring appeal is not merely its ability to capture the subtle dance between observed phenomena and their hidden causes, but the sheer elegance of its computational solution. At its heart lies a daunting challenge: the number of possible hidden state sequences grows as \(d^{n+1}\) for \(d\) states and length \(n\), a combinatorial explosion that would rapidly become intractable. Yet, the Viterbi algorithm, a masterpiece of dynamic programming, tames this complexity, reducing what seems like an exponentially scaling problem to a disciplined \(O(d^2 n)\) computation \cite{V03, FV05, lou95}. This efficiency—the ability to extract the single most probable explanation from a sea of possibilities—is what transformed HMMs from a theoretical curiosity into a practical cornerstone of modern machine learning and signal processing \cite{Alali2024, WhiteMahonyBrushe00, Mitchell1997}. 

However, as we stand on the cusp of the quantum computing era, this very notion of scaling is being re‑evaluated. The emergence of quantum convolutional coding and its maximum‑likelihood decoding strategies, inspired by Viterbi‑type recursions in the quantum regime \cite{OllivierTillichPRL2003, OllivierTillich2006, TanLi2010}, together with explicitly quantum algorithms for Viterbi decoding of classical convolutional codes \cite{GriceMeyer2015}, signals a profound shift: dynamic programming on trellises is no longer confined to classical computation. These developments, synthesized and systematized in the modern theory of quantum convolutional codes \cite{WildeQECBook2013}, challenge us to ask whether the algorithmic guarantees that have underpinned the classical Viterbi algorithm for generations can be not only matched but fundamentally re‑engineered in a quantum setting. Can a quantum formulation of Viterbi decoding reshape the scaling landscape for inference over hidden processes, moving beyond the limits of classical dynamic programming to approach, or even attain, genuine quantum advantage in sequential data analysis?
Quantum Markov chains (QMCs) provide a principled non‑commutative generalization of classical Markov chains, replacing stochastic kernels by completely positive, identity‑preserving transition expectations on operator algebras \cite{accardi1974,ASS20,Ib08}. Their Hilbert–Schmidt duals are quantum channels, yielding the usual Schrödinger‑picture formulation familiar in quantum information theory \cite{FR15,D18,Ib08}. In finite dimensions, a QMC is determined by a transition expectation \(\mathcal{E}:\mathcal{B}(\mathcal{H})\otimes\mathcal{B}(\mathcal{H})\to\mathcal{B}(\mathcal{H})\) that is completely positive, unital, and defines a consistent family of local states on the quasi‑local algebra \(\mathcal{B}(\mathcal{H})^{\otimes\mathbb{N}}\), thereby encoding a genuinely quantum Markov property \cite{ASS20}. When \(\mathcal{E}\) preserves a diagonal subalgebra, the construction collapses to a classical Markov chain on a finite alphabet, but in general the presence of off‑diagonal coherence leads to dynamics that have no classical counterpart \cite{Lu95,ASS20,Ib08}.

Hidden quantum Markov models (HQMMs) enrich this picture by coupling a hidden quantum chain to an observable system via quantum instruments, yielding a quantum analogue of classical HMMs \cite{AGLS24Q,Monras11}. In the HQMP formulation, one specifies hidden and observable tensor‑product algebras, a family of hidden transition expectations \(\mathcal{E}_{H;n}\), and emission maps \(\mathcal{E}_{H,O;n}\) that implement the measurement–back‑action mechanism on the hidden algebra, with a joint state \(\varphi_{H,O}\) enforcing consistency of the process \cite{AGLS24Q}. Classical HMMs are recovered when these maps preserve diagonal subalgebras, while the fully quantum case allows coherent hidden evolution and nontrivial back‑action, producing temporal correlations beyond any classical HMM of the same hidden dimension \cite{AGLS24Q,Monras11}. 

Moreover, HQMMs connect naturally to entangled hidden Markov models and matrix product state descriptions of quantum processes, where output statistics are encoded by MPS generated from suitable Kraus operators \cite{SS23,Sou25,FLW24}. This embeds HQMMs into the tensor‑network toolkit of quantum many‑body physics and clarifies their capacity to model finite‑memory quantum stochastic processes \cite{Monras11,Sou25,FLW24}. Taken together, these constructions provide a robust mathematical framework within which one can formulate and analyze rigorous quantum analogues of classical decoding algorithms—most notably Viterbi‑type and Bellman‑recursive schemes for quantum and hidden quantum Markov models—thereby offering a natural setting to pursue the quantum decoding questions raised above \cite{ASS20,AGLS24Q,Monras11,Sou25}.

\subsection{Contribution}

We develop a Viterbi-type decoding scheme for hidden quantum Markov models (HQMMs) in an operator-algebraic setting, placing quantum decoding within the non-commutative Markov and HQMP framework introduced above \cite{ASS20,AGLS24Q,Monras11}. Given a finite observation sequence \((q_0,\dots,q_n)\) with \(q_m\in\mathcal{P}_1(\mathcal{K})\) (rank-one projections on \(\mathcal{K}\)), the quantum Viterbi problem is to select a hidden trajectory \((p_0,\dots,p_n)\) of pure effects \(p_m\in\mathcal{P}_1(\mathcal{H})\subset\mathsf{Eff}(\mathcal{H})\) that maximizes the joint score functional
\[
\psi_n(p_0,\dots,p_n)
:=
\varphi_{H,O}\!\left(\bigotimes_{m=0}^n (p_m\otimes q_m)\right),
\]
where \(\mathsf{Eff}(\mathcal{H})=\{e\in\mathcal{B}(\mathcal{H})\mid 0\le e\le \mathbb{I}\}\) is ordered by \(A\le B \Leftrightarrow B-A\ge 0\), and \(\mathcal{P}_1(\mathcal{H})\) plays the role of an “atomic’’ hidden state space, generalizing the finite state set \(S\) of a classical HMM \cite{ASS20,AGLS24Q}. The basic optimization task is to compute
\[
\psi_{\mathrm{quantum}}^*
:=
\sup_{(p_0,\dots,p_n)\in\mathcal{P}_1(\mathcal{H})^{n+1}}
\psi_n(p_0,\dots,p_n)
\]
and to identify an optimal or approximately optimal trajectory \((p_0^*,\dots,p_n^*)\).

In the classical case, with hidden states \(i_m\in S\), this supremum reduces to a finite maximum
\[
\psi_{\mathrm{classical}}^*
=
\max_{(i_0,\dots,i_n)\in S^{n+1}}
\psi_n\bigl(|i_0\rangle\!\langle i_0|,\dots,|i_n\rangle\!\langle i_n|\bigr),
\]
and the Markov property yields a scalar factorization into transition and emission probabilities, which underlies the standard dynamic-programming (Viterbi) recursion with \(O(N^2 n)\) complexity for \(|S|=N\) \cite{Rab2002,V03,FV05}. Our construction lifts this recursion to the operator level: transition and emission probabilities are replaced by completely positive, identity-preserving maps and quantum instruments \cite{accardi1974,Ib08,AGLS24Q}, and the recursion propagates operator-valued “scores’’ on \(\mathcal{P}_1(\mathcal{H})\) rather than scalars.

We then specialize to a physically motivated quantum memory model with hidden qubit \(\mathcal{H}\cong\mathbb{C}^2\) and two-dimensional observation space \(\mathcal{K}\cong\mathbb{C}^2\), whose dynamics interpolate between coherent unitary rotations and phase-damping channels, while emissions are implemented by weak measurements that partially distinguish computational-basis states but preserve coherence \cite{FR15,D18}. In this setting we prove a strict decoding-level quantum advantage: there exist finite horizons \(n\) and observation sequences \((q_0,\dots,q_n)\) such that
\[
\psi_{\mathrm{quantum}}^*
=
\sup_{(p_0,\dots,p_n)\in\mathcal{P}_1(\mathcal{H})^{n+1}} \psi_n(p_0,\dots,p_n)
\;>\;
\psi_{\mathrm{classical}}^*
=
\sup_{(p_0,\dots,p_n)\in D_H^{n+1}} \psi_n(p_0,\dots,p_n),
\]
where \(D_H=\{|0\rangle\langle0|,|1\rangle\langle1|\}\subset\mathcal{P}_1(\mathcal{H})\) is the diagonal “classical’’ subset. Thus, even for a two-dimensional hidden space, the HQMM decoder attains scores unattainable by any classical HMM with the same cardinality, provided the recursion is allowed to explore coherent (off-diagonal) effects, echoing known representational gains of quantum models for finite-memory processes \cite{Monras11,FLW24,Sou25}. Geometrically, the advantage stems from the non-commutative structure of \(\mathcal{P}_1(\mathcal{H})\): our recursion optimizes over the full Bloch sphere, whereas a classical decoder is confined to its poles.

Beyond this canonical qubit memory example, the same framework applies to more general quantum signal-processing and sequence-modeling tasks, where decoding must be performed in the presence of temporal quantum correlations—such as quantum filtering, quantum trajectory estimation, or quantum-enhanced sequence learning \cite{ticozzi2008,HY2018,SSP15,J22}. In this sense, our operator-algebraic quantum Viterbi scheme is not only a theoretical generalization of the classical algorithm, but also a practical template for designing decoding routines in HQMM-based models of quantum memories, quantum communication channels, and quantum machine-learning architectures.

\subsection{Quantum Viterbi landscape and outlook}

A number of works have developed Viterbi-type decoding schemes within quantum information and coding theory. In the setting of quantum convolutional codes, maximum-likelihood decoders based on dynamic programming yield “quantum Viterbi’’ algorithms whose complexity scales linearly with the number of encoded qubits, thereby providing a direct quantum analogue of classical trellis-based decoding in the stabilizer formalism \cite{OllivierTillichPRL2003,OllivierTillich2006,WildeQECBook2013}. Subsequent developments have focused on improving the efficiency of such decoders: hybrid constructions embed carefully engineered classical Viterbi subroutines into quantum-coding architectures, significantly reducing the computational overhead of syndrome-based maximum-likelihood decoding \cite{TanLi2010}, while explicit Schr\"odinger-picture quantum Viterbi algorithms for classical convolutional codes demonstrate how superposition and amplitude amplification can be leveraged to obtain asymptotic speedups under suitable structural assumptions \cite{GriceMeyer2015}. Beyond coding theory, related ideas appear in the study of quantum stochastic processes and hidden quantum dynamics: HQMMs and quantum filtering frameworks provide a natural setting for sequence estimation in the presence of quantum memory effects \cite{ticozzi2008,HY2018,SSP15,J22}, and recent advances in characterizing quantum memory via channel discrimination and process-tensor methods suggest structural connections between decoding functionals and operational witnesses of non-Markovianity \cite{TaElM,OZNPQ26}.

Against this backdrop, the operator-algebraic, Heisenberg-picture formulation of the quantum Viterbi recursion introduced in this work opens several directions for further research. A primary objective is the development of concrete implementations, including circuit-level realizations and hybrid architectures in which HQMM-based quantum decoders interface with quantum convolutional and stabilizer-code Viterbi schemes \cite{OllivierTillich2006,TanLi2010,GriceMeyer2015}. On the algorithmic side, scalability remains a central challenge, motivating efficient approximation strategies such as adaptive refinements of $\delta$-nets on $\mathcal{P}_1(\mathcal{H})$, variational parametrizations of effects tailored to specific HQMM structures, and tensor-network representations of the recursion to control complexity in large hidden dimensions $N$ \cite{GT25,Sou25}. From an applications perspective, the strict quantum advantage we establish suggests that HQMM-based decoders may yield performance gains for tasks governed by temporal quantum correlations, including quantum filtering, quantum trajectory estimation, and quantum-enhanced sequence modeling in machine learning \cite{ticozzi2008,HY2018,SSP15,J22}. In addition, the connection between Viterbi-type functionals and quantum memory witnesses indicates that such functionals could serve as task-oriented diagnostics for hidden quantum memory, complementing existing process-tensor and channel-based criteria \cite{TaElM,OZNPQ26}. Finally, the structural parallel between our recursion and classical dynamic programming points toward hybrid classical–quantum optimization schemes, where quantum interference and amplitude amplification are exploited to approximate the argmax over $\mathcal{P}_1(\mathcal{H})^{n+1}$, in line with emerging proposals for quantum optimization and interferometric decoders \cite{Ambainis2003,JSW2025,SCFFK25}. Within the broader HQMM framework—particularly with the newly introduced class of causal HQMMs—these developments are expected to extend the applicability of quantum Viterbi methods across a wide range of models and domains.

In the remainder of the paper, we develop this program in several steps. In Section~\ref{sect-preli}, we review the operator-algebraic formulation of quantum Markov chains and hidden quantum Markov models, emphasizing the transition expectations, quantum instruments, and structural conditions that connect the HQMM framework to classical HMMs. Section~\ref{Quantum_Vit_Alg} introduces the quantum Viterbi decoding scheme, formulating path optimization directly on manifolds of quantum effects and deriving an operator-valued dynamic recursion that generalizes the classical Bellman–Viterbi equations to the non-commutative setting. In Section~\ref{Sect_Class}, we analyze the classical limit of our construction and establish strict quantum optimality results, showing that HQMM-based decoders can outperform all classical Viterbi decoders with the same hidden cardinality. Section~\ref{sect_App} applies the theory to a concrete quantum memory model, illustrating how quantum Viterbi decoding reveals advantages in decoding performance and memory representation for qubit-based processes. Finally, Section~\ref{Sect_disc} discusses broader implications for quantum error correction, quantum sequence modeling, and future directions in quantum dynamic programming and decoding.

\section{Preliminaries on Hidden Quantum Markov  Models}\label{sect-preli}

Let \(\mathcal{H}\) be a finite-dimensional Hilbert space of dimension \(N\), representing the  hidden state space . It is equipped with an orthonormal basis \(\{|j\rangle \mid j \in I_H\}\) indexed by \(I_H = \{1, \dots, N\}\). The algebra of all bounded linear operators on \(\mathcal{H}\), denoted \(\mathcal{B}_H = \mathcal{B}(\mathcal{H})\), captures all possible transformations of the hidden quantum system.

Dually, let \(\mathcal{K}\) be an \(M\)-dimensional Hilbert space representing the  observation space , with orthonormal basis \(\{e_k \mid k \in I_O\}\) and \(I_O = \{1, \dots, M\}\). The corresponding observation algebra is \(\mathcal{B}_O = \mathcal{B}(\mathcal{K})\). At each discrete time step \(k \in \mathbb{N}\), we consider copies \(\mathcal{H}_k \cong \mathcal{H}\) and \(\mathcal{K}_k \cong \mathcal{K}\), and define the local  hidden-observable algebra  for that instant as
\[
\mathcal{B}_{H,O;k} := \mathcal{B}_{H;k} \otimes \mathcal{B}_{O;k},
\]
where \(\mathcal{B}_{H;k} \cong \mathcal{B}_H\) and \(\mathcal{B}_{O;k} \cong \mathcal{B}_O\). The tensor product structure formally encodes the coupling between the latent quantum system and the measurement apparatus at time \(k\).

For a finite time horizon \([0,n] := \{0,1,\dots,n\}\), we construct the temporal algebras that describe the entire history:
\[
\mathcal{B}_{H;[0,n]} = \bigotimes_{k=0}^{n} \mathcal{B}_{H;k}, \quad
\mathcal{B}_{O;[0,n]} = \bigotimes_{k=0}^{n} \mathcal{B}_{O;k}, \quad
\mathcal{B}_{H,O;[0,n]} = \bigotimes_{k=0}^{n} \mathcal{B}_{H,O;k}.
\]
The corresponding  local algebras , defined as the inductive limits
\[
\mathcal{B}_{H;\text{loc}} := \bigcup_{n \ge 0} \mathcal{B}_{H;[0,n]}, \quad
\mathcal{B}_{O;\text{loc}} := \bigcup_{n \ge 0} \mathcal{B}_{O;[0,n]}, \quad
\mathcal{B}_{H,O;\text{loc}} := \bigcup_{n \ge 0} \mathcal{B}_{H,O;[0,n]},
\]
consist of all operators with finite temporal support. The full  sample algebra  is the C\(^*\)-completion
\[
\mathcal{B}_{H,O;\mathbb{N}} = \overline{\bigotimes_{k \in \mathbb{N}} (\mathcal{B}_{H;k} \otimes \mathcal{B}_{O;k})}^{\|\cdot\|},
\]
which provides the mathematical universe for defining quantum stochastic processes over infinite time, generalizing the classical concept to the non-commutative setting \cite{Neumann18, KemenySnell83}.

\begin{definition}[Hidden Quantum Markov Process (HQMP)]\label{def-QHMP}
A state \(\varphi_{H,O} : \mathcal{B}_{H,O;\mathbb{N}} \to \mathbb{C}\) is called a  hidden quantum Markov process  if it is completely determined by:
\begin{itemize}
    \item An  initial state  \(\phi_{H,0} : \mathcal{B}_H \to \mathbb{C}\),
    \item  Hidden transition maps  \(\{\mathcal{E}_{H;n} : \mathcal{B}_{H;n} \otimes \mathcal{B}_{H;n+1} \to \mathcal{B}_{H;n}\}_{n \ge 0}\) (completely positive, unital),
    \item  Emission maps  \(\{\mathcal{E}_{H,O;n} : \mathcal{B}_{H;n} \otimes \mathcal{B}_{O;n} \to \mathcal{B}_{H;n}\}_{n \ge 0}\) (completely positive).
\end{itemize}
For any finite sequence \(a_m \in \mathcal{B}_{H;m}, b_m \in \mathcal{B}_{O;m}\) (\(m = 0,\dots,n\)):
\begin{equation}\label{eq:joint-expectation-recursive}
\begin{aligned}
\varphi_{H,O}\left(\bigotimes_{m=0}^{n} (a_m \otimes b_m)\right) = \phi_{H,0} \Bigg( \mathcal{E}_{H;0}\bigg( \mathcal{E}_{H,O;0}(a_0 \otimes b_0) \otimes \mathcal{E}_{H;1}\Big( \mathcal{E}_{H,O;1}(a_1 \otimes b_1) \otimes \\
\cdots \otimes \mathcal{E}_{H;n}\big( \mathcal{E}_{H,O;n}(a_n \otimes b_n) \otimes \mathbb{I}_{H;n+1} \big) \cdots \Big) \bigg) \Bigg).
\end{aligned}
\end{equation}
Equivalently, defining  block maps  \(\mathcal{F}_{a,b}^{(n)}(\cdot) := \mathcal{E}_{H;n}\big(\mathcal{E}_{H,O;n}(a \otimes b) \otimes \cdot\big)\):
\[
\varphi_{H,O}\left(\bigotimes_{m=0}^{n} (a_m \otimes b_m)\right)
= \phi_{H,0} \circ \mathcal{F}_{a_0,b_0}^{(0)} \circ \mathcal{F}_{a_1,b_1}^{(1)} \circ \cdots \circ \mathcal{F}_{a_n,b_n}^{(n)}(\mathbb{I}_{H;n+1}).
\]
The triple \(\varphi_{H,O} \equiv (\phi_{H,0}, (\mathcal{E}_{H;n}), (\mathcal{E}_{H,O;n}))\) constitutes a  hidden quantum Markov chain .
\end{definition}

The algebraic framework underlying Hidden Quantum Markov Processes exhibits a structure naturally amenable to quantum circuit interpretation, wherein the information exchange between hidden and observable subsystems is systematized through complementary encoder-decoder operations.   Given a hidden system state $a_n \in \mathcal{B}_{H;n}$ and an observation outcome encoded by the operator $b_n \in \mathcal{B}_{O;n}$, the combined dynamical evolution through one temporal step follows the pattern:

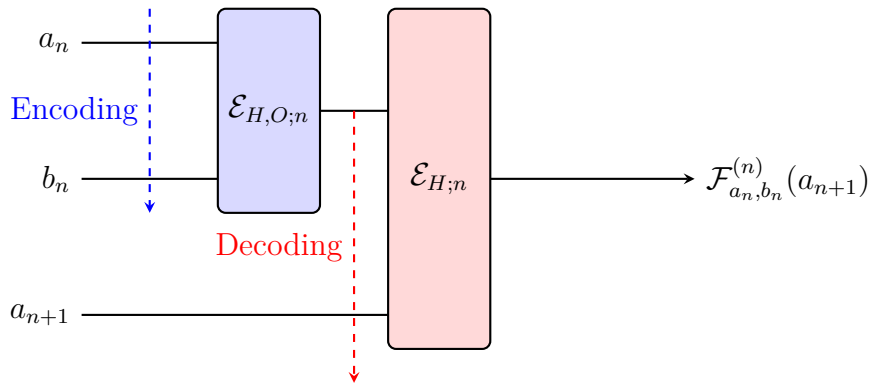
\begin{center}
\begin{tikzpicture}[thick, scale=0.9, >=stealth]
    % Registers / Wires
    \draw (0,2) -- (2,2); % a_n to Encoder
    \draw (0,0) -- (2,0); % b_n to Encoder
    \draw (3.5,1) -- (4.5,1); % Encoder Output to Decoder
    \draw (0,-2) -- (4.5,-2); % rhq_{n+1} to Decoder

    % Labels for inputs
    \node[left] at (0,2) {\(a_n\)};
    \node[left] at (0,0) {\(b_n\)};
    \node[left] at (0,-2) {\(a_{n+1}\)};

    % Encoder (Emission) - Process a_n and b_n
    \draw[fill=blue!15, rounded corners=3pt] (2,-0.5) rectangle (3.5,2.5);
    \node at (2.75,1) {\(\mathcal{E}_{H,O;n}\)};

    % Decoder (Transition) - Centered label and single output
    \draw[fill=red!15, rounded corners=3pt] (4.5,-2.5) rectangle (6,2.5);
    \node at (5.25,0) {\(\mathcal{E}_{H;n}\)};

    % Final Output - Single line from the middle of the red rectangle
    \draw[->] (6,0) -- (9,0) node[right] {\(\mathcal{F}_{a_n,b_n}^{(n)}(a_{n+1})\)};

    % Arrows showing flow
    \draw[->, dashed, blue] (1,2.5) -- (1,-0.5) node[midway,left] {Encoding};
    \draw[->, dashed, red] (4,1) -- (4,-3 ) node[midway,left] {Decoding};

\end{tikzpicture}
\end{center}

The emission operation $\mathcal{E}_{H,O;n} : \mathcal{B}_{H;n} \otimes \mathcal{B}_{O;n} \to \mathcal{B}_{H;n}$ functions as an encoding channel: it synthesizes information from the hidden state $a_n \in \mathcal{B}_{H;n}$ and the observed outcome $b_n \in \mathcal{B}_{O;n}$, producing an intermediate effective state residing in $\mathcal{B}_{H;n}$ that encapsulates the measurement-induced disturbance. Employing the Kraus decomposition $\{\mathcal{K}_i^{(n)}\}_{i \in \mathbb{N}}$ widely used in quantum information theory~\cite{Nielsen2000}, one may express this as:
\begin{equation}
\mathcal{E}_{H,O;n}(a_n \otimes b_n) = \sum_{i \in \mathbb{N}} K_i^{(n)} (a_n \otimes b_n) (K_i^{(n)})^\dagger,
\end{equation}
wherein the Kraus operators $\{K_i^{(n)}\}_{i \in \mathbb{N}}$ satisfy the completeness relation $\sum_{i \in \mathbb{N}} (K_i^{(n)})^\dagger K_i^{(n)} = \mathbb{I}_{H;n}$. Such completely positive, trace-preserving (CPTP) maps constitute the natural generalization of classical stochastic transition kernels to the quantum regime.

Conversely, the transition map $\mathcal{E}_{H;n} : \mathcal{B}_{H;n} \otimes \mathcal{B}_{H;n+1} \to \mathcal{B}_{H;n}$, which is completely positive and unital, serves a decoding role: it processes the encoded information together with the subsequent hidden state $a_{n+1} \in \mathcal{B}_{H;n+1}$, yielding a result in $\mathcal{B}_{H;n}$ that characterizes the evolution of the hidden subsystem over the temporal interval $[n, n+1]$. In the language of open quantum systems and dynamical semigroups~\cite{ticozzi2008}, this operation may be expressed via the Stinespring dilation as a partial trace over an ancillary environmental space $\mathcal{H}_E$:
\begin{equation}
\mathcal{E}_{H;n}(a_n \otimes a_{n+1}) = \text{Tr}_{E}\left[U^{(n)}(a_n \otimes a_{n+1} \otimes a_E^{(n)})(U^{(n)})^\dagger\right],
\end{equation}
for an environmental state $a_E^{(n)} \in \mathcal{B}(\mathcal{H}_E)$ and a unitary operator $U^{(n)} : \mathcal{H}_{H;n} \otimes \mathcal{H}_{H;n+1} \otimes \mathcal{H}_E \to \mathcal{H}_{H;n} \otimes \mathcal{H}_{H;n+1} \otimes \mathcal{H}_E$ that governs the joint evolution.

  For fixed observation $b_n \in \mathcal{B}_{O;n}$ and hidden operator $a_n \in \mathcal{B}_{H;n}$, applied to an input state $a_{n+1} \in \mathcal{B}_{H;n+1}$, it is defined as:
\begin{equation}
\mathcal{F}_{a_n,b_n}^{(n)}(a_{n+1}) := \mathcal{E}_{H;n}\left(\mathcal{E}_{H,O;n}(a_n \otimes b_n) \otimes a_{n+1}\right).
\end{equation}
This composition inherits the CPTP property from its constituent maps, thereby preserving quantum positivity and trace structure at each step.

The mathematical equivalence between the Schrödinger picture---which describes the evolution of density operators \(\rho, \sigma \in \mathcal{D}(\mathcal{H})\)---and the Heisenberg picture---which governs observables \(O, A \in \mathcal{B}(\mathcal{H})\)---is established via the  Hilbert--Schmidt duality. This duality ensures that for any quantum operation \(\mathcal{E}\), the adjoint map \(\mathcal{E}^{\dagger}\) satisfies the trace relation
\[
\operatorname{Tr}\!\big[\mathcal{E}(\rho) A\big] = \operatorname{Tr}\!\big[\rho \; \mathcal{E}^{\dagger}(A)\big],
\]
for all states \(\rho\) and observables \(A\). As a result, HQMMs apply equally to \textbf{state-based estimation problems}, such as quantum filtering. This algebraic framework, built on operator algebras \(\mathcal{B}_{H,O;\mathbb{N}}\), positive operator-valued measures, and completely positive maps, provides a systematic foundation for constructing quantum generalizations of classical inference algorithms---most notably the quantum Viterbi algorithm, which is elaborated in later sections of this work~\cite{FV05, V03, ASS20}.

\section{Quantum Viterbi Decoding: Path Optimization on Quantum State Manifolds}\label{Quantum_Vit_Alg}
We now establish the rigorous mathematical framework for quantum path optimization within the context of  HQMMs \cite{Monras11, AGLS24Q}. Let $\mathcal{H}$ and $\mathcal{K}$ be finite-dimensional complex Hilbert spaces representing the hidden and observation systems, respectively, with $\dim\mathcal{H} = N$. The space of density operators on $\mathcal{H}$ is the convex set
\[
\mathcal{D}(\mathcal{H}) := \{\rho \in \mathcal{B}(\mathcal{H}) \mid \rho \geq 0,\ \operatorname{tr}(\rho) = 1\}.
\]
Its extreme boundary consists precisely of the rank-one projections, which correspond to pure quantum states \cite{Neumann18, Nielsen2000}.

The space of \emph{pure quantum states} on the hidden Hilbert space $\mathcal{H}$ is characterized by
\[
\mathcal{P}_1(\mathcal{H}) = \left\{ p \in \mathcal{B}(\mathcal{H}) \mid p = p^\dagger = p^2,\ \operatorname{tr}(p) = 1 \right\} = \left\{ |\xi\rangle\langle\xi| : \xi \in \mathcal{H},\ \|\xi\| = 1 \right\}.
\]
This space is isomorphic to the complex projective space $\mathbb{C}P^{N-1}$ and inherits a canonical   structure via the Fubini--Study metric \cite{Nielsen2000}. Analogously, the observation space admits the pure state space
\[
\mathcal{P}_1(\mathcal{K}) = \left\{ q \in \mathcal{B}(\mathcal{K}) \mid q = q^\dagger = q^2,\ \operatorname{tr}(q) = 1 \right\} = \left\{ \chi\chi^* : \chi \in \mathcal{K},\ \|\chi\| = 1 \right\}.
\]

In the HQMM framework, $\mathcal{P}_1(\mathcal{H})$ parametrizes possible hidden quantum trajectories, while $\mathcal{P}_1(\mathcal{K})$ represents potential measurement outcomes. This generalizes the state-space formalism of classical hidden Markov models \cite{Rab2002, WhiteMahonyBrushe00} to the quantum regime, where the hidden dynamics are governed by a quantum stochastic process \cite{accardi1974, ASS20}. The symplectic geometry of $\mathcal{P}_1(\mathcal{H})$ plays a fundamental role in the analysis of quantum dynamics and control \cite{ticozzi2008, GT25}.

In the standard formulation of quantum measurement theory, an effect  on the Hilbert space \(\mathcal{H}\) is a bounded self-adjoint operator \(E \in \mathcal{B}(\mathcal{H})\) satisfying \(0 \leq E \leq \mathbb{I}\), representing a yes–no measurement outcome with outcome probabilities given by \(\operatorname{tr}(\rho E)\) for a state \(\rho\) \cite{MilP07}. The set of all effects on \(\mathcal{H}\) is denoted by
\begin{equation}\label{eq:effects}
\mathsf{Eff}(\mathcal{H}) := \{\, T \in \mathcal{B}(\mathcal{H}) \mid 0 \leq T \leq \mathbb{I} \,\}.
\end{equation}

and forms, with the inherited L\"owner order, a prototypical example of an effect algebra that captures the unsharp, probabilistic structure of quantum events \cite{MilP07}.

A core algorithmic challenge in extending classical decoding methods to the quantum setting is the need to optimize over continuous families of projection-valued paths. The classical Viterbi algorithm \cite{V03, FV05, Shi2021} relies on maximization over discrete paths; its quantum analogue requires suprema of families of quantum effects (operators between $0$ and $\mathbb{I}$). The following lemma guarantees the existence of such suprema, providing the necessary order-theoretic foundation.

\subsection{Existence of Optimal Hidden Quantum Paths}
\label{subsec:optimal-existence}

 QVA addresses the following  estimation problem in  HQMMs: given a finite temporal sequence of observable (pure) states \(\{q_m \in \mathcal{P}_1(\mathcal{K})\}_{m=0}^n\), identify the optimal sequence of hidden state projections \(\{p_m^* \in \mathcal{P}_1(\mathcal{H})\}_{m=0}^n\) that maximizes the joint quantum expectation functional. Formally, the problem is
\begin{equation}\label{eq:qviterbi-problem}
\max_{p_0,\dots,p_n \in \mathcal{P}_1(\mathcal{H})} \varphi_{H,O}\left(\bigotimes_{m=0}^{n} (p_m \otimes q_m)\right),
\end{equation}
where \(\varphi_{H,O}\) denotes the expectation functional of the underlying hidden quantum Markov process (QHMP) \cite{ASS20, AGLS24Q}. This constitutes the quantum generalization of classical Viterbi decoding \cite{V03, FV05}, where optimization now occurs over continuous, projection-valued paths inhabiting the complex projective manifold \(\mathcal{P}_1(\mathcal{H}) \cong \mathbb{C}P^{N-1}\), rather than a finite classical state set.

The existence of a solution to \eqref{eq:qviterbi-problem} is non-trivial, as it involves searching within the infinite, yet compact, geometric space of pure quantum states. This search inherits the ambient Riemannian geometry of the complex projective space, which is endowed with the Fubini–Study metric \cite{Nielsen2000}. The optimization problem thus takes place on a finite Cartesian product of compact   manifolds, a setting where global existence theorems from variational analysis on manifolds apply \cite{ticozzi2008}. The following theorem provides the rigorous guarantee for the existence of optimal quantum trajectories, a prerequisite for the well-posedness of the quantum Viterbi algorithm \cite{Li2024, MRDFS22}.

\begin{theorem}[Existence of Optimal Quantum Trajectories]\label{thm:optimal-sequence-existence}
For any finite observation sequence \(\{q_m\}_{m=0}^n\) with \(q_m \in \mathcal{P}_1(\mathcal{K})\), there exists an optimal sequence of hidden states \((p_0^*, p_1^*, \dots, p_n^*)\) with \(p_m^* \in \mathcal{P}_1(\mathcal{H})\) such that
\[
\varphi_{H,O}\left(\bigotimes_{m=0}^{n} (p_m^* \otimes q_m)\right) = \sup_{p_0,\dots,p_n \in \mathcal{P}_1(\mathcal{H})} \varphi_{H,O}\left(\bigotimes_{m=0}^{n} (p_m \otimes q_m)\right).
\]
Moreover, the supremum is attained as a maximum, and the search for this optimum is performed over the compact product manifold
\[
\mathcal{X}_n := \prod_{m=0}^{n} \mathcal{P}_1(\mathcal{H}) \cong (\mathbb{C}P^{N-1})^{n+1}.
\]
\end{theorem}

\begin{proof}
The proof relies on three interconnected facts: the compactness of the pure-state manifold, the continuity of the joint expectation functional, and the topological structure of finite products of compact spaces.

First, consider the search space. The set of rank-one projections \(\mathcal{P}_1(\mathcal{H})\) is norm-closed and bounded within \(\mathcal{B}(\mathcal{H})\); its identification with the complex projective space \(\mathbb{C}P^{N-1}\) via the map \(\xi \mapsto |\xi\rangle\langle\xi|\) (for \(\|\xi\|=1\)) endows it with the structure of a compact, connected   manifold \cite{Nielsen2000}. Since we consider a finite time horizon \(n+1\), the Cartesian product
\[
\mathcal{X}_n := \prod_{m=0}^{n} \mathcal{P}_1(\mathcal{H})
\]
is compact in the product topology (by Tychonoff's theorem). This topology is equivalent to that induced by any product metric, such as the sum of Fubini–Study distances, which makes \(\mathcal{X}_n\) a compact metric space \cite{GeSunLee01, ticozzi2008}.

Second, define the objective functional
\[
\psi_n: \mathcal{X}_n \longrightarrow \mathbb{R}, \qquad \psi_n(p_0,\dots,p_n) := \varphi_{H,O}\!\left(\bigotimes_{m=0}^{n} (p_m \otimes q_m)\right).
\]

We establish the continuity of \(\psi_n\). The map
\[
\iota: \mathcal{X}_n \longrightarrow \bigotimes_{m=0}^{n} \bigl(\mathcal{B}(\mathcal{H}) \otimes \mathcal{B}(\mathcal{K})\bigr), \quad (p_0,\dots,p_n) \mapsto \bigotimes_{m=0}^{n} (p_m \otimes q_m)
\]
is continuous because, in finite dimensions, the tensor product operation and the multiplication by a fixed operator \(q_m\) are jointly continuous in the norm topology. Furthermore, \(\varphi_{H,O}\) is a linear functional on the finite-dimensional C$^*$-algebra generated by the process up to time \(n\) \cite{ASS20}. Every linear functional on a finite-dimensional normed space is bounded, hence (norm-)continuous. Consequently, \(\psi_n = \varphi_{H,O} \circ \iota\) is a composition of continuous maps and is therefore continuous on \(\mathcal{X}_n\).

Finally, by the extreme value theorem (applicable to continuous real-valued functions on compact topological spaces), a continuous function on a compact set attains its supremum. Since \(\mathcal{X}_n\) is compact and \(\psi_n\) is continuous, there exists a tuple \((p_0^*,\dots,p_n^*) \in \mathcal{X}_n\) such that
\[
\psi_n(p_0^*,\dots,p_n^*) = \max_{(p_0,\dots,p_n)\in\mathcal{X}_n} \psi_n(p_0,\dots,p_n).
\]
This proves the existence of an optimal hidden quantum path.
\end{proof}

\begin{remark}
Theorem \ref{thm:optimal-sequence-existence} guarantees that the quantum decoding problem is well-posed, but it does not prescribe a constructive method. The quantum Viterbi algorithm \cite{Li2024, MRDFS22} aims to compute this optimum through a quantum dynamical programming recursion, analogous to its classical counterpart \cite{Rab2002, FV05}. The search occurs within the ambient geometry of \(\mathcal{X}_n \cong (\mathbb{C}P^{N-1})^{n+1}\), a manifold whose curvature  influences gradient-based or geodesic search strategies. This geometric perspective connects to optimal control on homogeneous spaces   and to the simulation of quantum walks on complex projective spaces \cite{Ambainis2003, ApersSarletteTicozzi18}. The finiteness of the time horizon is crucial; it reduces an infinite-dimensional path optimization to a finite-dimensional variational problem on a compact manifold, enabling the use of global optimization techniques from Riemannian geometry \cite{GT25}.
\end{remark}

\begin{remark}
In classical HMMs, the Viterbi algorithm searches over a finite set of discrete state sequences, where existence of a maximum is trivial and complexity is polynomial in the sequence length \cite{V03, Rab2002}. The quantum version replaces the finite set with the continuous manifold \(\mathcal{X}_n\). While Theorem \ref{thm:optimal-sequence-existence} assures existence, the computational complexity of finding the global maximum on a high-dimensional   manifold is generally non-trivial. This reflects a fundamental shift in the nature of decoding: from combinatorial optimization to continuous, geometric optimization on a curved space. Recent works on model reduction for quantum processes \cite{WhiteMahonyBrushe00, GT25} and spectral learning of HQMMs \cite{HsuKakadeZhang12, Srin2017} aim to tame this complexity by exploiting the low-rank or symmetric structure of the underlying quantum Markov process.
\end{remark}

% \section{The Quantum Viterbi Algorithm: Dynamic Programming Principle}\label{Sect:QVA_Dyn}

\subsection{Quantum Viterbi Principle: Optimal Hidden Paths via Backward Optimization}

Theorem~\ref{thm:optimal-sequence-existence} ensures the existence of an optimal
hidden trajectory \((p_0^*,\dots,p_n^*)\) that maximizes the joint functional
\(\psi_n\). This provides the variational foundation for a Viterbi-type
backward-induction scheme in the quantum setting.

Since \(\psi_n\) is built from a nested composition of quantum operations, the
optimization problem is naturally solved from the final time step backwards.
At each stage \(m\), the optimal choice \(p_m^*\) depends on the already
optimized future segment \((p_{m+1}^*,\dots,p_n^*)\), and this dependence is
propagated backwards until the initial time. The compactness of
\(\mathcal{P}_1(\mathcal{H})\) guarantees that each intermediate maximization
problem admits a solution.

Formally, we are interested in
\begin{equation}\label{eq:psi-backward}
\begin{aligned}
\psi_n(p_0^*,\dots,p_n^*)
&= \max_{(p_0,\dots,p_n)\in\mathcal{X}_n} \psi_n(p_0,\dots,p_n) \\
&= \max_{(p_0,\dots,p_n)\in\mathcal{X}_n}
   \varphi_{H,0}\Bigl(
   \mathcal{F}^{(0)}_{p_0,q_0}\bigl(
   \mathcal{F}^{(1)}_{p_1,q_1}\bigl(
   \cdots
   \mathcal{F}^{(n)}_{p_n,q_n}(I)
   \cdots
   \bigr)
   \bigr)
   \Bigr) \\
&= \max_{p_0\in\mathcal{P}_1(\mathcal{H})}
   \max_{p_1\in\mathcal{P}_1(\mathcal{H})}
   \cdots
   \max_{p_n\in\mathcal{P}_1(\mathcal{H})}
   \varphi_{H,0}\Bigl(
   \mathcal{F}^{(0)}_{p_0,q_0}\bigl(
   \mathcal{F}^{(1)}_{p_1,q_1}\bigl(
   \cdots
   \mathcal{F}^{(n)}_{p_n,q_n}(I)
   \cdots
   \bigr)
   \bigr)
   \Bigr),
\end{aligned}
\end{equation}
where \(\mathcal{X}_n := \mathcal{P}_1(\mathcal{H})^{\times(n+1)}\). By
Theorem~\ref{thm:optimal-sequence-existence}, each maximization over
\(\mathcal{P}_1(\mathcal{H})\) is well-defined, so the iterated maximization in
\eqref{eq:psi-backward} makes sense.

\medskip

From a procedural point of view, we construct the optimal path by backward
selection of conditional optimizers. For fixed prefix \((p_0,\dots,p_{n-1})\),
define
\[
\widetilde{p}_n(p_0,\dots,p_{n-1})
\in \operatorname{argmax}_{p_n\in\mathcal{P}_1(\mathcal{H})}
\psi_n(p_0,\dots,p_{n-1},p_n),
\]
so that
\[
\max_{p_n\in\mathcal{P}_1(\mathcal{H})}
\psi_n(p_0,\dots,p_{n-1},p_n)
=
\psi_n\bigl(p_0,\dots,p_{n-1},\widetilde{p}_n(p_0,\dots,p_{n-1})\bigr).
\]
Intuitively, \(\widetilde{p}_n\) assigns to each history \((p_0,\dots,p_{n-1})\)
a best possible terminal hidden state at time \(n\).

We then move one step backward and, for fixed prefix
\((p_0,\dots,p_{n-2})\), optimize at stage \(n-1\) while taking into account
that time \(n\) will be chosen optimally:
\[
\widetilde{p}_{n-1}(p_0,\dots,p_{n-2})
\in \operatorname{argmax}_{p_{n-1}\in\mathcal{P}_1(\mathcal{H})}
\psi_n\bigl(p_0,\dots,p_{n-2},p_{n-1},
           \widetilde{p}_n(p_0,\dots,p_{n-2},p_{n-1})\bigr),
\]
so that
\[
\max_{p_{n-1}\in\mathcal{P}_1(\mathcal{H})}
\psi_n\bigl(p_0,\dots,p_{n-2},p_{n-1},
           \widetilde{p}_n(p_0,\dots,p_{n-2},p_{n-1})\bigr)
=
\psi_n\bigl(p_0,\dots,p_{n-2},
           \widetilde{p}_{n-1}(p_0,\dots,p_{n-2}),
           \widetilde{p}_n(\cdots)\bigr).
\]

Continuing inductively, we define for each \(m=n-2,\dots,1,0\) a measurable
selection
\[
\widetilde{p}_m(p_0,\dots,p_{m-1})
\in \operatorname{argmax}_{p_m\in\mathcal{P}_1(\mathcal{H})}
\psi_n\bigl(
p_0,\dots,p_{m-1},p_m,
\widetilde{p}_{m+1}(p_0,\dots,p_m),
\dots,
\widetilde{p}_n(p_0,\dots,p_{n-1})
\bigr),
\]
thereby propagating the optimal future decisions backward to the initial time.

At the end of this backward pass, we obtain the optimal initial hidden state
by solving the one-dimensional problem
\[
p_0^* \in \operatorname{argmax}_{p_0\in\mathcal{P}_1(\mathcal{H})}
\psi_n\bigl(
p_0,
\widetilde{p}_1(p_0),
\dots,
\widetilde{p}_n(p_0,\widetilde{p}_1(p_0),\dots,\widetilde{p}_{n-1}(\cdots))
\bigr).
\]
The associated Viterbi trajectory is then reconstructed forward via the
recursion
\[
p_1^* := \widetilde{p}_1(p_0^*),\quad
p_2^* := \widetilde{p}_2(p_0^*,p_1^*),\quad
\dots,\quad
p_n^* := \widetilde{p}_n(p_0^*,\dots,p_{n-1}^*).
\]

By construction, \((p_0^*,\dots,p_n^*)\) realizes the global maximum of
\(\psi_n\) over \(\mathcal{X}_n\). In other words, this backward-induction
procedure implements a quantum analogue of the classical Viterbi algorithm:
local optimization at each stage, conditioned on an already optimal “future”,
yields the globally optimal hidden quantum path consistent with the observed
sequence.
\begin{algorithm}[H]
\caption{Quantum Viterbi via Backward Selectors}
\begin{algorithmic}[1]
\State \textbf{Input:}
\[
\text{HQMM } 
(\varphi_{H,0},(\mathcal{E}_{H;m})_{m=0}^n,(\mathcal{E}_{H,O;m})_{m=0}^n),
\quad
\text{Observation sequence } (q_0,\dots,q_n)\in\mathcal{P}_1(\mathcal{K})^{n+1}.
\]
\State \textbf{Output:} $(p_0^*,\dots,p_n^*)\in\mathcal{P}_1(\mathcal{H})^{n+1}$, $\psi_n(p_0^*,\dots,p_n^*)$
\State

\State \textbf{Global functional:}
\[
\psi_n(p_0,\dots,p_n)
:= \varphi_{H,O}\!\Bigl(\bigotimes_{m=0}^n (p_m \otimes q_m)\Bigr).
\]

\State \textbf{Backward selectors:}
\State For $m=n$:
\[
\widetilde{p}_n(p_0,\dots,p_{n-1})
\in \arg\max_{p_n\in\mathcal{P}_1(\mathcal{H})}
      \psi_n(p_0,\dots,p_{n-1},p_n).
\]
\State For $m=n-1$:
\[
\widetilde{p}_{n-1}(p_0,\dots,p_{n-2})
\in \arg\max_{p_{n-1}\in\mathcal{P}_1(\mathcal{H})}
      \psi_n\bigl(p_0,\dots,p_{n-2},p_{n-1},
                  \widetilde{p}_n(p_0,\dots,p_{n-2},p_{n-1})\bigr).
\]
\State For $m<n$:
\[
\widetilde{p}_m(p_0,\dots,p_{m-1})
\in \arg\max_{p_m\in\mathcal{P}_1(\mathcal{H})}
      \psi_n\bigl(p_0,\dots,p_{m-1},p_m,
                  \widetilde{p}_{m+1}(p_0,\dots,p_m),
                  \dots,
                  \widetilde{p}_n(p_0,\dots,p_{n-1})\bigr).
\]

\State \textbf{Initial maximizer:}
\[
F(p_0)
:= \psi_n\bigl(p_0,
               \widetilde{p}_1(p_0),
               \dots,
               \widetilde{p}_n(p_0,\dots,\widetilde{p}_{n-1}(\cdots))\bigr),
\quad
p_0^* \in \arg\max_{p_0\in\mathcal{P}_1(\mathcal{H})} F(p_0).
\]

\State \textbf{Forward reconstruction:}
\[
p_1^* := \widetilde{p}_1(p_0^*),\quad
p_2^* := \widetilde{p}_2(p_0^*,p_1^*),\ \dots,\ 
p_n^* := \widetilde{p}_n(p_0^*,\dots,p_{n-1}^*).
\]

\State \textbf{Optimal value:}
\[
\psi_n(p_0^*,\dots,p_n^*)
= \max_{(p_0,\dots,p_n)\in\mathcal{P}_1(\mathcal{H})^{n+1}} \psi_n(p_0,\dots,p_n).
\]

\end{algorithmic}
\end{algorithm}

\section{Classical Limit and Strict Quantum Optimality in Viterbi Decoding}\label{Sect_Class}

\subsection{Recovering Classical Viterbi Decoding}
 
A fundamental consistency requirement for any quantum generalization of a classical algorithm is that it must reduce to its classical counterpart under appropriate structural constraints. We establish here the precise conditions under which the quantum Viterbi algorithm collapses to the classical Viterbi algorithm \cite{V03, Rab2002}, thereby demonstrating that our quantum framework properly extends classical hidden Markov theory.

\begin{theorem}\label{thm:quantum-classical-reduction}
Let $\mathcal{H}$ be an $N$-dimensional Hilbert space with distinguished orthonormal bases $\{|h_j\rangle\}_{j\in D}$ for the hidden system and $\{|q_k\rangle\}_{k\in D}$ for observations, where $D = \{1,\dots,N\}$. Consider a hidden quantum Markov process satisfying:
\begin{enumerate}
    \item The hidden transition operator is diagonal in the hidden basis:
    \[
    \mathcal{E}_H(X) = \sum_{i,j\in D} p_{H;ij} \langle h_i|X|h_j\rangle |h_i\rangle\langle h_j|,
    \]
    where $P_H = (p_{H;ij})$ is a stochastic matrix.

    \item The emission operator is diagonal in both bases:
    \[
    \mathcal{E}_{O,H}(X\otimes Y) = \sum_{i,j\in D} p_{O|H}(j|i) \langle h_i|X|h_i\rangle \langle q_j|Y|q_j\rangle |h_i\rangle\langle h_i|,
    \]
    where $p_{O|H}(\cdot|\cdot)$ defines observation probabilities conditioned on hidden states.

    \item The initial state is diagonal: $\phi_{H,0}(X) = \sum_{i\in D} p^{(0)}_i \langle h_i|X|h_i\rangle$.
\end{enumerate}
Then for any observation sequence $(q_1,\dots,q_n) = (|q_{k_1}\rangle\langle q_{k_1}|, \dots, |q_{k_n}\rangle\langle q_{k_n}|)$, the quantum Viterbi algorithm computes exactly the classical Viterbi score:
\[
\max_{p_1,\dots,p_n \in \mathcal{P}_1(\mathcal{H})} \varphi_{H,O}\!\left(\bigotimes_{m=1}^n (p_m \otimes q_m)\right)
= \max_{i_1,\dots,i_n \in D} p^{(0)}_{i_1} \prod_{m=1}^{n-1} p_{H;i_m i_{m+1}} \prod_{m=1}^n p_{O|H}(k_m|i_m).
\]
Moreover, any optimal quantum trajectory must be of the form $p_m^* = |h_{i_m^*}\rangle\langle h_{i_m^*}|$ for some classical optimal path $(i_1^*,\dots,i_n^*)$.
\end{theorem}

\begin{proof}
The proof proceeds by establishing an isomorphism between the diagonal subalgebra and classical probability. Let $\mathsf{Eff}(\mathcal{H})_{\text{diag}} = \operatorname{span}\{|h_i\rangle\langle h_i| : i \in D\} \subset \mathcal{B}(\mathcal{H})$ denote the commutative $C^*$-algebra of diagonal operators. By assumptions (1)-(3), the entire quantum hidden Markov process restricts to $\mathsf{Eff}(\mathcal{H})_{\text{diag}}$.

Consider the recursion of Theorem \ref{thm:optimal-sequence-existence}. Under the diagonal emission assumption, for any $p = |h_i\rangle\langle h_i|$ and $q = |q_k\rangle\langle q_k|$, we have:
\[
\mathcal{E}_{O,H}(p \otimes q) = p_{O|H}(k|i) |h_i\rangle\langle h_i|.
\]
The transition operator, when applied to diagonal inputs $X = \sum_i \alpha_i |h_i\rangle\langle h_i|$ and $Y = \sum_j \beta_j |h_j\rangle\langle h_j|$, yields:
\[
\mathcal{E}_H(X \otimes Y) = \sum_{i,j} p_{H;ij} \alpha_i \beta_j |h_i\rangle\langle h_i|.
\]
Thus, if we define classical value functions $v_m: D \to \mathbb{R}$ by $v_m(i) = \phi_{H,0}(V_m(|h_i\rangle\langle h_i|))$, they satisfy the classical Viterbi recursion:
\[
v_n(i) = p_{O|H}(k_n|i), \quad
v_m(i) = \max_{j\in D} p_{H;ij} p_{O|H}(k_m|i) v_{m+1}(j).
\]
The solution of this recursion is precisely the classical Viterbi algorithm \cite{V03, FV05}.

To show that no non-diagonal projection can yield a higher score, note that the initial functional $\phi_{H,0}$ and all transition/emission maps are completely positive and preserve diagonal observable. For any projection $p \in \mathcal{P}_1(\mathcal{H})$, let $p = \sum_{i,j} c_i\overline{c_j} |h_i\rangle\langle h_j|$ be with $\sum_i |c_i|^2 = 1$. Then:
\[
\mathcal{E}_{O,H}(p \otimes |q_k\rangle\langle q_k|) = \sum_i |c_i|^2 p_{O|H}(k|i) |h_i\rangle\langle h_i|.
\]
This is a convex combination of diagonal operators. By the concavity of the classical recursion, the maximum must occur at an extreme point of the set of such convex combinations, i.e., at some $|c_i| = 1$, $c_j = 0$ for $j \neq i$. Hence, the optimal quantum projections must be diagonal.
\end{proof}

Theorem~\ref{thm:quantum-classical-reduction} establishes that, under assumptions of natural diagonalization in the channels and initial state, the quantum Viterbi recursion collapses exactly to the classical Viterbi algorithm in a finite hidden state space. This provides the essential consistency check: whenever the underlying dynamics are effectively commutative, our HQMM formalism does not invent new behavior, but faithfully reproduces the standard hidden Markov decoding theory.

\section{Application to Quantum Memory Advantages}\label{sect_App}
\label{sec:spin-ensemble}
We illustrate the quantum Viterbi framework in a concrete setting motivated by quantum memory models. The hidden system is a single effective qubit, described at each discrete time \(m \in \mathbb{N}\) by the two-dimensional Hilbert space \(\mathcal{H} = \mathbb{C}^2\) and observable algebra \(\mathcal{B}_{H;m} = \mathcal{B}(\mathbb{C}^2) \cong M_2(\mathbb{C})\).  This minimal setup already realizes the basic architecture used in quantum implementations of hidden Markov models, where a finite-dimensional quantum system simulates a classical stochastic process while potentially using less memory than any classical HMM with the same predictive capabilities \cite{Elliot21}. The system may be implemented, for example, as a spin-\(\tfrac{1}{2}\) particle, a two-level atom. The observation space is likewise two-dimensional, \(\mathcal{K} = \mathbb{C}^2\), with orthonormal basis \(\{|+\rangle,|-\rangle\}\) representing binary measurement outcomes. Hidden configurations at time \(m\) are pure effects \(p_m \in \mathcal{P}_1(\mathcal{H}) \subset \mathsf{Eff}(\mathcal{H})\), while observations are pure effects \(q_m \in \mathcal{P}_1(\mathcal{K})\), so the quantum Viterbi recursion of Theorem~\ref{thm:optimal-sequence-existence} applies directly to this qubit memory model.

We specify the hidden dynamics as a driven, partially dephasing qubit. Let
\[
\sigma_x
=
\begin{pmatrix}
0 & 1\\
1 & 0
\end{pmatrix}
=
|0\rangle\langle 1| + |1\rangle\langle 0|
\]
denote the Pauli-\(X\) operator. For each time \(m\), define the unitary
\[
U_m := e^{-i\phi_m\sigma_x}
  = \cos\phi_m \, \mathbb{I} - i \sin\phi_m \, \sigma_x,
\]
with \(\phi_m \in \mathbb{R}\). This is the standard single-qubit rotation about the \(x\)-axis of the Bloch sphere and generates coherent Rabi oscillations between \(|0\rangle\) and \(|1\rangle\), as in driven two-level systems used to implement and manipulate elementary quantum memories \cite{Nielsen2000}. Such coherent gates are exactly the type of unitary updates used in quantum stochastic simulators, where a low-dimensional quantum system realizes a process as a sequence of controlled quantum operations \cite{ElliotiGu24}.

To model loss of phase coherence, we introduce the \(Z\)-dephasing channel \(\Pi_Z : \mathcal{B}(\mathbb{C}^2)\to\mathcal{B}(\mathbb{C}^2)\) given by
\[
\Pi_Z(a) = |0\rangle\langle0|\,a\,|0\rangle\langle0|
          + |1\rangle\langle1|\,a\,|1\rangle\langle1|,
\]
which removes off-diagonal entries in the computational basis and is the canonical phase-damping channel for two-level memories \cite{Banchi24}. At time \(m\), we interpolate between coherent and dephasing evolution via
\[
\Phi_m^{\mathrm{S}}(a)
  = (1-\theta_m)\,U_m a U_m^\dagger + \theta_m\,\Pi_Z(a),
  \qquad \theta_m \in [0,1],
\]
and, in the Heisenberg picture, the dual map on effects
\[
\Phi_m(X)
  = (1-\theta_m )\,U_m^\dagger X U_m + \theta_m \,\Pi_Z(X).
\]
The parameter \(\theta_m \) thus interpolates between fully coherent evolution (\(\theta_m  = 0\)) and fully dephased dynamics (\(\theta_m  = 1\)). This directly mirrors models used to quantify how much coherence and hence quantum memory can be retained in finite-dimensional quantum simulators of stochastic processes under noisy evolution \cite{Elliot21}.

\subsection{Coherent hidden dynamics with controlled decoherence}

Within the hidden quantum Markov framework, a one-step hidden transition is encoded by a completely positive, identity-preserving map \(\mathcal{E}_{H;m}:\mathcal{B}_{H;m}\otimes\mathcal{B}_{H;m+1}\to\mathcal{B}_{H;m}\). For the present model we define
\[
\mathcal{E}_{H;m}(X_1\otimes X_2) := \Phi_m(X_1)\,\mathrm{tr}(X_2),
\]
where \(\mathrm{tr}\) is the normalized trace on \(M_2(\mathbb{C})\).

\begin{lemma}
The map \(\mathcal{E}_{H;m} : \mathcal{B}_{H;m}\otimes\mathcal{B}_{H;m+1} \to \mathcal{B}_{H;m}\) defined above is completely positive and identity-preserving, hence a valid quantum transition expectation in the sense of operator-algebraic quantum Markov chains.
\end{lemma}

\begin{proof}
Unitary conjugation \(X\mapsto U_m^\dagger X U_m\) is a *-automorphism of \(\mathcal{B}(\mathcal{H})\), so it is completely positive and unital. The dephasing map \(\Pi_Z\) is also completely positive and unital, since it can be written with Kraus operators \(|0\rangle\langle0|\) and \(|1\rangle\langle1|\). Therefore their convex combination \(\Phi_m\) is completely positive and unital. The map \(\mathrm{id}\otimes\mathrm{tr}\) acts as the identity on the first factor and as a state on the second, so it is completely positive and sends \(\mathbb{I}\otimes \mathbb{I}\) to \(\mathbb{I}\). The composition \(\mathcal{E}_{H;m} = \Phi_m \circ (\mathrm{id}\otimes\mathrm{tr})\) is completely positive and satisfies
\[
\mathcal{E}_{H;m}(\mathbb{I}\otimes \mathbb{I})
  = \Phi_m(\mathrm{tr}(\mathbb{I})\,\mathbb{I})
  = \Phi_m(\mathbb{I})
  = \mathbb{I},
\]
using unitality of \(\Phi_m\). This proves the claim.
\end{proof}

This transition expectation realizes a hidden quantum Markov chain on the effect set \(\mathsf{Eff}(\mathcal{H})\), interpolating between a genuinely quantum Markov evolution (for small \(p_m\)) and a classical Markov chain on the diagonal subalgebra (for \(p_m = \mathbb{I}\)). In particular, it fits into the general framework of hidden quantum Markov processes where the hidden dynamics live on a non-commutative operator algebra and the observed process is obtained by applying suitable instruments \cite{AGLS24Q}.

\begin{remark}
The two-level setting is the simplest regime in which the non-commutative geometry of the effect set \(\mathsf{Eff}(\mathcal{H})\) becomes operationally relevant. The manifold \(\mathcal{P}_1(\mathcal{H})\) already contains coherent superpositions such as \(|+\rangle\langle+|\) that lie outside the commutative subalgebra generated by \(|0\rangle\langle0|\) and \(|1\rangle\langle1|\). This allows one to realize processes whose hidden dynamics cannot be captured by any classical Markov chain on two states and connects directly with the identification of quantum memory as a resource that can reduce the effective state-space dimension needed to simulate a given process \cite{Elliot21}. The same mechanism underlies more general trade-offs between accuracy and memory in higher-dimensional quantum simulations of stochastic processes, where one balances the amount of coherence retained in the hidden system against the achievable fidelity of the simulated output statistics \cite{Banchi24}.
\end{remark}

\subsection{Quantum measurement with coherent back-action}

We now specify the observation mechanism, which couples the hidden memory qubit to the two-dimensional measurement space \(\mathcal{K}\) and produces classical outputs while feeding back on the hidden effect. In the HQMM formalism this is described by a completely positive map \(\mathcal{E}_{H,O;m}:\mathcal{B}_{H;m}\otimes\mathcal{B}_{O;m}\to\mathcal{B}_{H;m}\) acting as a quantum instrument.

Fix a measurement strength parameter \(0<\eta<1\) and define
\[
L^{+}
  = \sqrt{1+\eta}\,|0\rangle\langle0|
    + \sqrt{1-\eta}\,|1\rangle\langle1|,
\qquad
L^{-}
  = \sqrt{1-\eta}\,|0\rangle\langle0|
    + \sqrt{1+\eta}\,|1\rangle\langle1|.
\]
A direct computation shows
\[
(L^{+})^\dagger L^{+} + (L^{-})^\dagger L^{-} = 2I,
\]
so that, up to normalization, \(\{L^{+},L^{-}\}\) define a two-outcome measurement on the hidden qubit. We associate \(L^{+}\) with the observation effect \(q^{(+)} = |+\rangle\langle+|\) and \(L^{-}\) with \(q^{(-)} = |-\rangle\langle-|\). On elementary tensors we set
\[
\mathcal{E}_{H,O;m}(A\otimes q^{(+)})
  = (L^{+})^\dagger A L^{+},
\qquad
\mathcal{E}_{H,O;m}(A\otimes q^{(-)})
  = (L^{-})^\dagger A L^{-},
\]
for \(A\in\mathcal{B}_{H;m}\), and extend linearly to \(\mathcal{B}_{H;m}\otimes\mathcal{B}_{O;m}\). This admits the Kraus form
\[
\mathcal{E}_{H,O;m}(X)
  = (W^{+})^\dagger X W^{+} + (W^{-})^\dagger X W^{-},
\qquad
W^{\pm} := L^{\pm}\otimes|\pm\rangle,
\]
and thus defines a completely positive quantum instrument in the usual sense \cite{Nielsen2000}.

The single parameter \(\eta\) interpolates between non-informative and projective measurements. When \(\eta=0\), one has \(L^{+} = L^{-} = \mathbb{I}\), so the emission leaves all hidden effects unchanged and carries no information about the hidden state. As \(\eta\to 1\), the operators \(L^{\pm}\) approach projectors onto \(|0\rangle\) and \(|1\rangle\), and the emission becomes a projective measurement in the computational basis, which is maximally informative but fully de-coheres the hidden qubit. For intermediate values \(0<\eta<1\), the measurement is weak: it partially distinguishes the basis states while preserving coherence in the hidden effect, exactly the regime relevant for hidden quantum memory effects under imperfect measurements \cite{Monras11,TaElM,TaBra25}.

To see the effect on the hidden effect set explicitly, let \(p_m = |\psi_m\rangle\langle\psi_m| \in \mathcal{P}_1(\mathcal{H})\) with
\[
|\psi_m\rangle = \alpha_m|0\rangle + \beta_m|1\rangle,
\qquad
|\alpha_m|^2 + |\beta_m|^2 = 1.
\]
A straightforward matrix computation in the basis \(\{|0\rangle,|1\rangle\}\) gives
\[
\mathcal{E}_{H,O;m}(p_m\otimes q^{(+)})
=
\begin{pmatrix}
(1+\eta)|\alpha_m|^2 & \sqrt{1-\eta^2}\,\alpha_m\overline{\beta_m} \\
\sqrt{1-\eta^2}\,\overline{\alpha_m}\beta_m & (1-\eta)|\beta_m|^2
\end{pmatrix}.
\]
The off-diagonal terms are proportional to \(\sqrt{1-\eta^2}\,\alpha_m\overline{\beta_m}\) and are nonzero whenever \(\alpha_m\beta_m\neq 0\) and \(\eta\neq 1\). Thus, for coherent inputs and weak measurement, the emission map preserves coherence in \(\mathsf{Eff}(\mathcal{H})\), maintaining access to non-commutative points of \(\mathcal{P}_1(\mathcal{H})\) that underpin quantum memory advantages in decoding \cite{TaElM,TaBra25,OZNPQ26}.
 
\subsection{Implications for quantum decoding superiority}

The hidden quantum Markov model defined above combines three structural ingredients that are crucial for quantum decoding advantages: coherent hidden states in \(\mathcal{P}_1(\mathcal{H})\), a tunable interplay between unitary rotations and dephasing in the transition expectations \(\mathcal{E}_{H;m}\), and a weak measurement instrument \(\mathcal{E}_{H,O;m}\) that partially extracts information while preserving coherence in the effect set. In combination with the quantum Viterbi recursion of Theorem~\ref{thm:optimal-sequence-existence}, this allows one to construct parameter regimes in which the supremum of the Viterbi functional over the full pure effect set \(\mathcal{P}_1(\mathcal{H})\) is strictly larger than its supremum over the classical subset of diagonal rank-one effects, even though both models use the same underlying Hilbert space and observation alphabet. This is in line with the general observation that quantum simulators of stochastic processes can achieve memory efficiencies or representational capabilities that are inaccessible to classical HMMs of the same nominal dimension \cite{Elliot21,Banchi24,ElliotiGu24,YangElliot25,Lu95,AGLS24Q}.

To make this more precise, fix a time index \(m=0\), an observation effect \(q_0 \in \mathcal{P}_1(\mathcal{K})\), and a positive operator \(V_1 \in \mathcal{B}_{H;1}\) that plays the role of a “continuation value” in the dynamic programming sense: it encodes the optimal future contribution to the decoding objective, conditional on optimally chosen hidden states at times \(1,\dots,n\). Let \(\varphi_{H,0}\) be an initial state on \(\mathcal{B}_{H;0}\). The one-step score functional is
\begin{equation}
\label{eq:one-step-functional}
F_0(p_0)
:=
\varphi_{H,0}\!\left(
\mathcal{E}_{H;0}\!\left(
\mathcal{E}_{H,O;0}(p_0\otimes q_0)\otimes V_1
\right)\right),
\qquad
p_0 \in \mathcal{P}_1(\mathcal{H}),
\end{equation}
which quantifies, within the full recursion, the contribution of choosing a given hidden pure effect \(p_0\) at time \(0\). In a classical HMM on a two-state space, one would restrict \(p_0\) to the finite set \(D_H = \{|0\rangle\langle0|,|1\rangle\langle1|\}\). The following theorem shows that, for suitable parameters of the qubit memory model, the supremum of \(F_0\) over the full pure effect set is strictly larger than its supremum over \(D_H\), and that any maximiser must be a coherent (off-diagonal) effect.

\begin{theorem}[Single-step quantum advantage]
\label{thm:coherent-optimal-m0}
There exist physically admissible parameters \((\eta_0,p_0,\phi_0)\), an observation \(q_0\), a continuation operator \(V_1\), and an initial functional \(\varphi_{H,0}\) such that the one-step functional \(F_0\) satisfies
\[
\sup_{p_0\in\mathcal{P}_{1}(\mathcal{H})} F_0(p_0)
\;>\;
\sup_{p_0\in D_H} F_0(p_0),
\]
where \(D_H=\{|0\rangle\langle0|,|1\rangle\langle1|\}\) denotes the classical subset of diagonal pure effects. Moreover, any hidden effect \(p_0^*\) achieving the supremum on the left must be coherent, in the sense that it has nonzero off-diagonal entries in the computational basis.
\end{theorem}

\begin{proof}
We argue by explicit construction. Choose parameters as follows: set dephasing \(\theta_0 = 0\) so that the hidden transition at time \(0\) is purely unitary; choose \(\phi_0 = \pi/4\), giving \(U_0 = e^{-i(\pi/4)\sigma_x}\); set the measurement strength \(\eta_0 = 1/2\), so that
\[
L_0^{+} = \sqrt{\tfrac{3}{2}}\,|0\rangle\langle0| + \sqrt{\tfrac{1}{2}}\,|1\rangle\langle1|;
\]
fix the observation to be \(q_0 = |+\rangle\langle+|\), where \(|+\rangle = (|0\rangle+|1\rangle)/\sqrt2\); choose the continuation operator to be \(V_1 =  I_2\); and take the initial functional \(\varphi_{H,0}(X) = \mathrm{Tr}(aX)\) with \(a = |\phi_0\rangle\langle\phi_0|\) and \(|\phi_0\rangle = (|0\rangle+|1\rangle)/\sqrt2\).

With \(\theta_0=0\) and \(\mathrm{Tr}(V_1)=1\) (normalized trace), the hidden transition simplifies to \(\mathcal{E}_{H;0}(X_1\otimes V_1) = U_0^\dagger X_1 U_0\). Using that \(U_0|\phi_0\rangle = e^{-i\pi/4}|\phi_0\rangle\), one finds \(U_0 a U_0^\dagger = a\). Consequently, the one-step functional reduces to
\[
F_0(p_0) = \mathrm{Tr}\bigl(a\,\mathcal{E}_{H,O;0}(p_0\otimes q_0)\bigr),
\]
and, since \(q_0 = |+\rangle\langle+|\) selects the \(L_0^{+}\) outcome, we have \(\mathcal{E}_{H,O;0}(p_0\otimes q_0) = L_0^{+} p_0 L_0^{+}\).

Let \(p_0 = |\psi\rangle\langle\psi|\) with \(|\psi\rangle = \alpha|0\rangle + \beta|1\rangle\) and \(|\alpha|^2+|\beta|^2=1\). A direct matrix calculation gives
\[
L_0^{+} p_0 L_0^{+}
  =
  \begin{pmatrix}
    \tfrac{3}{2}|\alpha|^2 & \tfrac{\sqrt{3}}{2}\,\alpha\overline{\beta} \\
    \tfrac{\sqrt{3}}{2}\,\overline{\alpha}\beta & \tfrac{1}{2}|\beta|^2
  \end{pmatrix}.
\]
With \(a = \tfrac12\begin{pmatrix}1&1\\[2pt]1&1\end{pmatrix}\), this yields
\[
F_0(p_0)
  = \tfrac{3}{8}|\alpha|^2 + \tfrac{1}{8}|\beta|^2
    + \tfrac{\sqrt{3}}{8}(\alpha\overline{\beta} + \overline{\alpha}\beta).
\]
On the classical subset \(D_H\), one finds \(F_0(|0\rangle\langle0|) = 3/8\) and \(F_0(|1\rangle\langle1|) = 1/8\), so
\[
\sup_{p_0\in D_H} F_0(p_0) = 3/8.
\]

Let
\[
x =
\begin{pmatrix}
\alpha \\
\beta
\end{pmatrix}
\in \mathbb{C}^2,
\quad \text{with } \|x\|^2 = |\alpha|^2 + |\beta|^2 = 1.
\]
Then \(F_0(p_0)\) can be written as a quadratic form
\[
F_0(p_0)=G(\alpha,\beta) = x^* A x,
\]
where
\[
A =
\frac{1}{8}
\begin{pmatrix}
3 & \sqrt{3} \\
\sqrt{3} & 1
\end{pmatrix}.
\]
The matrix \(A\) is Hermitian, hence it admits a spectral decomposition with real eigenvalues  \(1/2\) and \(0\). Therefore, for any unit vector \(x \in \mathbb{C}^2\),
\[
x^* A x = \lambda_1 |c_1|^2 + \lambda_2 |c_2|^2
\]
where \(\lambda_1,\lambda_2\) are the eigenvalues of \(A\), and
\[
|c_1|^2 + |c_2|^2 = 1.
\]
Thus \(x^* A x\) is a convex combination of the eigenvalues, and we obtain
\[
0
\le x^* A x
\le \tfrac{1}{2}.
\]
Consequently,
\[
\sup_{\|x\|=1} x^* A x = 1/2.
\]

This yields
\[
\sup_{|\alpha|^2 + |\beta|^2 = 1} G(\alpha,\beta) = \frac{1}{2}.
\]

For the coherent state
$$
|\psi^*\rangle
=
\frac{\sqrt{3}}{2}|0\rangle + \frac{1}{2}|1\rangle,
\qquad
p_0^* = |\psi^*\rangle\langle\psi^*|
$$
we have
\[
\sup_{p_0\in\mathcal{P}_1(\mathcal{H})} F_0(p_0)=F_0(p_0^*)
  = \tfrac{1}{2} > \sup_{p_0\in D_H} F_0(p_0) = 3/8.
\]

as claimed.

\end{proof}

\begin{theorem}[n-step quantum advantage]
\label{thm:coherent-optimal-m1}
There exist physically admissible parameters such that
\[
\sup_{p_0,p_1,\cdots,p_n\in\mathcal{P}_{1}(\mathcal{H})} F_1(p_0,p_1,\cdots, p_n)
\;>\;
\sup_{p_0,p_1, \cdots, p_n\in D_H} F_1(p_0,p_1,\cdots, p_n),
\] 
Moreover, any hidden effect \(p_0^*\) achieving the supremum on the left must be coherent, in the sense that it has nonzero off-diagonal entries in the computational basis.

\end{theorem}
\begin{proof}
We use the same parameters as in the single-step construction:
\[
\theta_i=0,\quad U_i=e^{-i(\pi/4)\sigma_x},\quad
q_i=|+\rangle\langle+|,\quad
L^+=\sqrt{\tfrac{3}{2}}|0\rangle\langle0|+\sqrt{\tfrac{1}{2}}|1\rangle\langle1|.
\]
Let \(a=|+\rangle\langle+|\) and \(V_n=I_2\).

With these choices, the functional factorizes as
\[
F_n(p_0,\dots,p_n)
=\mathrm{Tr}(aL^+p_0L^+)\prod_{i=1}^n \mathrm{Tr}(L^+p_iL^+).
\]

\medskip

\noindent

On the classical set \(D_H\), one computes
\[
\sup_{p_0\in D_H} \mathrm{Tr}(aL^+p_0L^+)=\tfrac{3}{8},
\quad
\sup_{p_i\in D_H} \mathrm{Tr}(L^+p_iL^+)=\tfrac{3}{4}\;\;(i\ge1).
\]
Hence,
\[
\sup_{p_0,\dots,p_n\in D_H} F_n
=
\tfrac{3}{8}\left(\tfrac{3}{4}\right)^n.
\]
Now, on  \(\mathcal{P}_1(H)\), one computes
\[
\sup_{p_0\in \mathcal{P}_1(\mathcal{H})} \mathrm{Tr}(aL^+ p_0^{\mathrm{coh}}L^+)=\tfrac{1}{2},
\quad
\sup_{p_i\in \mathcal{P}_1(\mathcal{H})} \mathrm{Tr}(L^+p_iL^+)=\tfrac{3}{4}\;\;(i\ge1).
\]

So the optimal path is the sequence
\[
|\psi^*\rangle
=
\frac{\sqrt{3}}{2}|0\rangle + \frac{1}{2}|1\rangle,
\qquad
p_0^* = |\psi^*\rangle\langle\psi^*|; \ p_1=\cdots=p_n=|0\rangle\langle0|.
\]

Thus,
\[
\sup_{p_0,\dots,p_n\in\mathcal{P}_1(\mathcal{H})} F_n=F_n(p_0^*,\dots,p_n)
=
\tfrac{1}{2}
\left(\tfrac{3}{4}\right)^n.
\]

\medskip

Since
\[
\tfrac{1}{2} > \tfrac{3}{8},
\]
we obtain
\[
F_n(p_0^*,\dots,p_n)
>
\sup_{p_0,\dots,p_n\in D_H} F_n.
\]
Therefore,
\[
\sup_{p_0,\dots,p_n\in\mathcal{P}_1(\mathcal{H})} F_n
>
\sup_{p_0,\dots,p_n\in D_H} F_n.
\]

\medskip
\end{proof}

\begin{remark}[Structure of the optimal trajectory]
\label{rem:hybrid-trajectory}
The proof reveals that the quantum advantage does not require coherence at all time steps. 
Instead, it is sufficient that coherence appears at a single step—here at the initial time \(0\). 
More precisely, the optimal strategy exhibits a \emph{hybrid structure}: the first hidden effect 
\(p_0\) is genuinely quantum (non-diagonal), while the subsequent effects 
\(p_1,\dots,p_n\) can be chosen among classical extremal points.

This phenomenon is a direct consequence of the multiplicative structure of the functional \(F_n\). 
The strict gain originates entirely from the first factor, where interference terms contribute 
positively, while the remaining factors are maximized by classical states. 

In particular, the quantum advantage is \emph{localized in time} but \emph{globally preserved} 
through the recursion. This shows that even a minimal injection of coherence into the hidden 
trajectory is sufficient to outperform any fully classical strategy.
\end{remark}

\subsection{Information-Theoretic Interpretation}

The quantum advantage established in Theorem~\ref{thm:coherent-optimal-m0} admits a natural interpretation in terms of information storage and processing in stochastic models. In classical hidden Markov models, memory is encoded as a probability distribution over a finite set of hidden states, which are mutually orthogonal and therefore perfectly distinguishable. As a consequence, classical memory can only represent uncertainty over discrete alternatives.

In contrast, quantum hidden Markov models encode memory in density operators (or, in our formulation, effects) on a Hilbert space. These states need not be orthogonal and can exhibit quantum coherence, allowing the model to store information not only in populations (diagonal entries) but also in phase relations (off-diagonal entries). This additional structure enables a more efficient encoding of predictive information.

This perspective is consistent with the theory of quantum simulators of stochastic processes, where quantum implementations can achieve either the same predictive performance with strictly less memory, or strictly better performance at fixed memory dimension, by exploiting coherence in internal states~\cite{Banchi24}.

Within this framework, we introduce the \emph{score gap}
\[
\Delta
=
\max_{p\in\mathcal{P}_1(\mathcal{H})} F(p)
-
\max_{p\in D_H} F(p),
\]
which quantifies the additional decoding performance enabled by quantum memory beyond the classical subset
\(
D_H = \{ |0\rangle\langle0|, |1\rangle\langle1| \}.
\)
A strictly positive value of $\Delta$ therefore provides an operational measure of quantum advantage.

In the single-step setting of Theorem~\ref{thm:coherent-optimal-m0}, the functional $F_0$ depends explicitly on the off-diagonal entries of the hidden effect. For a pure state
\[
p_0 = |\psi\rangle\langle\psi|, 
\qquad 
|\psi\rangle = \alpha|0\rangle + \beta|1\rangle,
\]
with $|\alpha|^2 + |\beta|^2 = 1$, one obtains
\[
F_0(p_0) - \sup_{p_0 \in D_H} F_0
=
\frac{\sqrt{3}}{4}\,\mathrm{Re}(\alpha\overline{\beta})
-
\frac{1}{4}|\beta|^2.
\]

This expression reveals two competing contributions. The first term,
\[
\mathrm{Re}(\alpha\overline{\beta}),
\]
originates from the off-diagonal entries of $p_0$ and captures quantum interference effects that are absent in classical models. The second term penalizes population imbalance and shows that not all quantum states yield an advantage.

In particular, the gain is directly controlled by the coherence of the state. Indeed, the $l_1$-coherence is given by
\[
C_{l_1}(p_0) = 2|\alpha\beta|,
\]
which quantifies the magnitude of the off-diagonal elements. However, the advantage is not determined solely by the amount of coherence, but also by its phase alignment, as only the real part $\mathrm{Re}(\alpha\overline{\beta})$ contributes positively. This shows that coherence acts as a resource only when it is properly aligned with the measurement structure.

From an information-theoretic viewpoint, this demonstrates that quantum models can encode predictively relevant information in phase relations between hidden states, rather than solely in classical probabilities. The resulting improvement does not arise from enlarging the hidden state space, but from exploiting non-commuting, coherent memory states within the same Hilbert space dimension.

Therefore, the score gap $\Delta$ provides an operational quantification of quantum coherence as a resource for inference: it isolates the contribution of off-diagonal degrees of freedom in enhancing decoding performance, thereby establishing a direct link between coherence and memory efficiency in quantum stochastic models~\cite{Elliot21}.

\begin{definition}[\(l_1\)-coherence]
Let \(\{|i\rangle\}\) be the computational basis of \(\mathcal{H}\). 
For any density operator \(\rho\) on \(\mathcal{H}\), the \(l_1\)-coherence of \(\rho\) is defined by
\[
C_{l_1}(\rho)
=
\sum_{i\neq j} \left| \langle i|\rho|j\rangle \right|.
\]
\end{definition}

\begin{proposition}[Coherence-dependent gain]
For the one-step functional \(F_0\) of Theorem~\ref{thm:coherent-optimal-m0}, the advantage
\[
\Delta_0
=
\max_{p_0\in\mathcal{P}_1(\mathcal{H})} F_0(p_0)
-
\max_{p_0\in D_H} F_0(p_0)
\]
satisfies
\begin{equation}\label{eq:gap-coherence-bound}
 \Delta_0
\leq
\frac{\sqrt{3}}{8}\,C_{l_1}(p_0^*)
-
\frac{1}{4}|\beta|^2   
\end{equation}

where $C_{l_1}(p_0^*) = 2|\alpha\beta|$ denotes the $l_1$-coherence. Moreover, $\Delta_0 > 0$ only if $C_{l_1}(p_0^*) > 0$, so coherence is a necessary resource for achieving a quantum advantage.
\end{proposition}

\begin{proof}
From the explicit computation of the one-step functional, we have
\[
F_0(p_0)
=
\frac{3}{8}|\alpha|^2
+
\frac{1}{8}|\beta|^2
+
\frac{\sqrt{3}}{8}(\alpha\overline{\beta} + \overline{\alpha}\beta).
\]
On the classical subset $D_H = \{|0\rangle\langle0|, |1\rangle\langle1|\}$, one computes
\[
\sup_{p_0 \in D_H} F_0(p_0) = \frac{3}{8}.
\]
Subtracting this value yields
\[
F_0(p_0) - \frac{3}{8}
=
\frac{\sqrt{3}}{4}\,\mathrm{Re}(\alpha\overline{\beta})
-
\frac{1}{4}|\beta|^2,
\]
Next, using the inequality
\[
\mathrm{Re}(\alpha\overline{\beta}) \le |\alpha\beta|,
\]
and the definition of the $l_1$-coherence
\(
C_{l_1}(p_0) = 2|\alpha\beta|,
\)
we obtain
\[
\mathrm{Re}(\alpha\overline{\beta})
\le
\frac{1}{2}C_{l_1}(p_0).
\]
So, one obtains
\[
\Delta
\le
\frac{\sqrt{3}}{8}\,C_{l_1}(p_0^*)
-
\frac{1}{4}|\beta|^2,
\]
which proves \eqref{eq:gap-coherence-bound}.

Finally, if $C_{l_1}(p_0) = 0$, then $\alpha\beta = 0$, so $p_0$ is diagonal and belongs to $D_H$. Hence $\Delta(p_0) \le 0$, showing that coherence is necessary for a strictly positive score gap.
\end{proof}

\begin{example}[Optimal coherent state and $n$-step amplification]
\label{ex:coherence-n-step}
Consider the coherent state
\[
|\psi^*\rangle
=
\frac{\sqrt{3}}{2}|0\rangle + \frac{1}{2}|1\rangle,
\qquad
p_0^* = |\psi^*\rangle\langle\psi^*|.
\]
A direct computation shows that
\[
F_0(p_0^*) = \frac{1}{2},
\qquad
\sup_{p_0 \in D_H} F_0(p_0) = \frac{3}{8},
\]
so that
\[
\Delta(p_0^*) = \frac{1}{8} > 0.
\]
In this case, the coherence is
\[
C_{l_1}(p_0^*) = 2 \left|\frac{\sqrt{3}}{2}\cdot \frac{1}{2}\right| = \frac{\sqrt{3}}{2},
\]
and the phase is optimally aligned so that
\(
\mathrm{Re}(\alpha\overline{\beta}) = |\alpha\beta|.
\)

\medskip

In the $n$-step setting of Theorem~\ref{thm:coherent-optimal-m1}, the functional factorizes as
\[
F_n(p_0,\dots,p_n)
=
\mathrm{Tr}(aL^+p_0L^+)\prod_{i=1}^n \mathrm{Tr}(L^+p_iL^+).
\]
Choosing
\[
p_0 = p_0^*,
\qquad
p_1 = \cdots = p_n = |0\rangle\langle0|,
\]
one obtains
\[
\sup_{\mathcal{P}_1(\mathcal{H})} F_n
=
\frac{1}{2}\left(\frac{3}{4}\right)^n,
\qquad
\sup_{D_H} F_n
=
\frac{3}{8}\left(\frac{3}{4}\right)^n.
\]
Therefore, the $n$-step score gap is
\[
\Delta_n
=
\left(\frac{1}{2} - \frac{3}{8}\right)\left(\frac{3}{4}\right)^n
=
\frac{1}{8}\left(\frac{3}{4}\right)^n > 0.
\]

\medskip

This example shows that the coherence-induced advantage at the initial step propagates multiplicatively along the optimal quantum path, demonstrating how quantum memory enhances decoding performance over multiple steps.
\end{example}

In the multi-step setting, Theorem~\ref{thm:optimal-sequence-existence} shows that the optimal Viterbi trajectory for a suitable finite horizon \(n\) must contain at least one coherent hidden state \(p_m^* \notin D_H\), and that the corresponding optimal quantum score
\[
\psi_{\mathrm{quantum}}^*
=
\varphi_{H,O}\!\left(\bigotimes_{m=0}^n (p_m^*\otimes q_m)\right)
\]
strictly exceeds the best possible classical score obtained by restricting all hidden states to \(D_H\). This demonstrates that, for the processes and observation sequences considered, no classical HMM—regardless of its parametrization—can reproduce the decoding performance of the quantum Viterbi algorithm. The situation is directly analogous to quantum stochastic simulators that require strictly fewer internal degrees of freedom than any classical model to reproduce the same process, or that achieve better approximation quality at a fixed memory dimension \cite{Banchi24,YangElliot25}.

The origin of this superiority is the same phenomenon identified as “hidden quantum memory” and “purely quantum memory” in recent work on temporal quantum correlations \cite{TaElM,TaBra25}. There, one finds processes whose observed statistics can be reproduced by classical models only if the internal memory is allowed to grow beyond what is needed in a quantum realization, or only if one discards certain temporal features. In our setting, the strict gap between quantum and classical Viterbi scores witnesses the presence of such quantum memory at the level of a decoding task: the process can be represented and decoded within a two-dimensional non-commutative effect space, but any classical description confined to diagonal states necessarily loses part of the information relevant for maximizing the trajectory score. This connects our decoding advantage to broader operational characterizations of quantum memory based on channel discrimination and constrained separability, where non-classical temporal correlations can be certified without enlarging the observed alphabet \cite{OZNPQ26}.

\section{Discussion}\label{Sect_disc}
The results developed in this work place the quantum Viterbi algorithm on a rigorous operator-algebraic footing and show that its optimization landscape can display a strict, structurally unavoidable advantage over any classical Viterbi decoder of the same nominal hidden dimension. At the core of this strict separation lies the optimization domain: the quantum recursion runs over the non-commutative effect set \(\mathsf{Eff}(\mathcal{H})\) and, more specifically, over the pure-effect manifold \(\mathcal{P}_1(\mathcal{H})\), whereas a classical HMM restricts to a finite set \(D_H\) of commuting rank-one projections. Theorem~\ref{thm:coherent-optimal-m0} shows that, for suitable HQMM data, there exist horizons \(n\) and observation sequences \((q_0,\dots,q_n)\) for which any trajectory maximizing the quantum Viterbi functional must include at least one coherent effect \(p_m^*\notin D_H\), and that the corresponding score strictly exceeds the supremum attainable when all hidden effects are constrained to lie in \(D_H\). This gap persists even though both descriptions share the same two-dimensional Hilbert space \(\mathcal{H}=\mathbb{C}^2\), highlighting that the relevant resource is not Hilbert-space dimension per se but the exploitation of non-commuting elements of \(\mathsf{Eff}(\mathcal{H})\).

From a quantum Markovian viewpoint, the construction fits naturally into the framework of operator-algebraic quantum Markov chains \cite{accardi1974} and hidden quantum Markov processes  \cite{AGLS24Q}, where transitions are encoded by completely positive, identity-preserving expectations \(\mathcal{E}_{H;m}\) on tensor products of local algebras and observations arise from quantum instruments \(\mathcal{E}_{H,O;m}\) acting on \(\mathcal{B}_{H;m}\otimes\mathcal{B}_{O;m}\). Our two-level memory model provides a minimal, yet nontrivial, instance of this theory: the hidden evolution \(\Phi_m\) interpolates between unitary rotations and dephasing, and the emission channel realizes a weak, binary measurement that partially distinguishes computational-basis effects while preserving off-diagonal coherence. In this setting, the strict inequality
\[
\sup_{(p_0,\dots,p_n)\in\mathcal{P}_1(\mathcal{H})^{n+1}}
\varphi_{H,O}\Bigl(\bigotimes_{m=0}^n (p_m\otimes q_m)\Bigr)
\;>\;
\sup_{(p_0,\dots,p_n)\in D_H^{n+1}}
\varphi_{H,O}\Bigl(\bigotimes_{m=0}^n (p_m\otimes q_m)\Bigr)
\]
should be interpreted as an operational witness that the underlying process cannot be “compressed” into a commutative hidden model of the same dimension without degrading decoding performance. In particular, it exhibits in a concrete decoding problem the type of quantum–classical representational gap that, on the level of stochastic simulation, underpins quantum memory compression and dimension reduction \cite{Elliot21,YangElliot25}.

In information-theoretic terms, the one-step functional \(F_0\) admits a closed expression that depends linearly on the off-diagonal entry of the hidden effect in the computational basis. For pure states \(p_0 = |\psi\rangle\langle\psi|\) with \(|\psi\rangle = \alpha|0\rangle + \beta|1\rangle\), the advantage \(\Delta_0\) over classical effects satisfies
\[
\Delta_0
=
\frac{\sqrt{3}-1}{4}\,C_{l_1}(p_0^*),
\]
where \(C_{l_1}(p_0^*) = \sum_{i\neq j}|\langle i|p_0^*|j\rangle|\) and \(p_0^*\) is an optimal state. Thus, in this example, the strict quantum advantage is directly proportional to a standard coherence monotone. This makes the role of coherence explicit: the improvement in Viterbi score stems precisely from the availability of hidden states with \(\langle 0|p_0^*|1\rangle\neq 0\), and vanishes if the process is constrained to a diagonal, classical sub-algebra. The structure mirrors results on quantum implementations of nondeterministic hidden Markov models, where non-orthogonal memory states enable a reduction of the statistical complexity relative to any classical simulator \cite{Elliot21}. It also parallels recent analyses of accuracy–memory trade-offs in quantum simulations of stochastic processes, where performance gains at fixed memory dimension are traced back to the ability to store predictive information in coherent, non commuting internal states rather than in classical mixtures \cite{Banchi24}.

The memory interpretation of this advantage connects our work to operational notions of “hidden quantum memory”  in time-correlated processes. In those settings, one considers families of temporal correlations or effective channels that admit a quantum realization with bounded internal dimension, but cannot be reproduced by any classical model of the same size without either enlarging the memory or discarding certain temporal features \cite{TaElM,TaBra25}. In our framework, the strict inequality between quantum and classical Viterbi scores plays an analogous role: it shows that, for the same observed process and fixed Hilbert-space dimension, restricting the hidden dynamics to classical effects eliminates part of the information relevant to the decoding task. One may view the fully quantum Viterbi decoder as saturating the representational capacity of a two-dimensional non commutative memory, while any classical restriction effectively underutilized that memory. From this perspective, the qubit HQMM studied here realises a simple instance of the “quantum dimensionality reduction” phenomena observed in more general sampling and simulation tasks, where a quantum device reproduces a target process using a hidden state space of smaller effective classical rank \cite{YangElliot25}.

The structural similarity between the quantum Viterbi recursion and classical dynamic programming suggests several directions for further work in quantum control and filtering. On the one hand, the recursion can be interpreted as a maximization over hidden trajectories in \(\mathcal{P}_1(\mathcal{H})\) of a value functional that factors through quantum transition expectations and instruments. This is close in spirit to quantum filtering and quantum trajectory theory, where one seeks to estimate conditional states of a system undergoing continuous or discrete-time measurements, and where non classical memory manifests in the dependence of future conditional states on past measurement records beyond any classical sufficient statistic \cite{ticozzi2008,HY2018}. A natural extension is to formulate quantum analogues of risk-sensitive or entropy-regularized Viterbi criteria, and to relate them to quantum versions of optimal filtering and smoothing. This would link our discrete-time, finite-horizon recursion to continuous-time stochastic master equations and to the broader control-theoretic framework of invariant and attractive quantum Markovian subsystems \cite{ticozzi2008}.

On the algorithmic side, the present work dovetails with efforts in quantum-enhanced machine learning and sequence modeling. HQMMs and related architectures have already been proposed as quantum generalizations of HMMs and probabilistic automata, with learning algorithms based on transition-operation matrices and tensor-network formalisms \cite{Monras11,Srin2017}. The strict decoding advantage exhibited here suggests that, even at fixed Hilbert-space dimension, quantum models may attain likelihoods or loss values on sequential data that are unattainable by any classical HMM of the same hidden cardinality. This raises the prospect of quantum Viterbi-style decoders serving as subroutines in quantum sequence models for tasks such as time-series prediction, anomaly detection, or structured sequence labeling, potentially complementing existing proposals for quantum-enhanced classification and generative modeling \cite{SSP15,J22}. An important open problem is to characterize, for realistic quantum hardware constraints, how the continuous optimization over \(\mathcal{P}_1(\mathcal{H})\) can be approximated efficiently—e.g.\ by variational families or discretizations—and how the resulting quantum models compare empirically to best-in-class classical baselines on real-world sequential data.

Finally, the dynamic-programming structure of the quantum Viterbi recursion suggests possible connections to recent work on quantum optimization and decoding algorithms that rely on interference-based exploration of path spaces. In particular, our recursion can be viewed as an abstract optimization over paths in a non-commutative state space, while proposals based on quantum walks and interferometric decoders achieve speedups or improved approximations by exploiting coherent superpositions of classical trajectories \cite{Ambainis2003,JSW2025}. Bridging these viewpoints—for instance, by embedding the Viterbi value recursion into a quantum circuit that approximates the "argmax" over \(\mathcal{P}_1(\mathcal{H})^{n+1}\) via amplitude amplification or adiabatic paths—could lead to hybrid classical–quantum decoding schemes that are both structurally interpretable and practically implementable \cite{CFG26}. In parallel, the link to quantum memory witnesses and channel-discrimination-based characterizations of memory points toward using Viterbi-type functionals as new, task-oriented diagnostics for hidden quantum memory, complementing existing channel- and process-tensor-based criteria \cite{OZNPQ26,TaElM}.

\end{document}